\documentclass[11pt]{article}

\usepackage[utf8]{inputenc}
\usepackage{t1enc}
\usepackage{amsmath,amstext,amscd,amsfonts,amsthm,stmaryrd}
\usepackage{amssymb}
\usepackage{geometry}
\usepackage{graphicx}
\usepackage{enumerate}
\usepackage{tikz}
\usetikzlibrary{arrows}
\usepackage{centernot,cancel}
\usepackage{authblk}
\usepackage{color}
\allowdisplaybreaks

\newtheorem{theorem}{Theorem}
\newtheorem{corollary}{Corollary}

\newtheorem{lemma}{Lemma}

\theoremstyle{remark}

\newcommand{\denselist}{\topsep 0pt\itemsep 0pt}

\newcommand{\tup}[1]{\langle #1 \rangle}

\renewcommand{\models}{\vDash}

\renewcommand{\L}{\mathcal{L}}

\newcommand{\Var}{\text{\bf Var}}

\newcommand{\antimodels}{$\mathrel{\raisebox{.008em}{\mbox{\rotatebox[origin=c]{180}{$\models$}}}}$}
\newcommand{\myequiv}{\mbox{\antimodels$\models$}}

\newcommand{\set}[1]{\left\{ #1 \right\}}
\newcommand{\setdef}[2]{\set{ #1 \, : \, #2}}
\newcommand{\seq}[2]{{#1}_0,{#1}_1, \dots, {#1}_{#2}}

\newcommand{\dans}{\longrightarrow}
\newcommand{\Set}[1]{\{#1\}}

\title{\bf Belief revision and 3-valued logics:\\
Characterization of 19,683 belief change operators}
\author{Nerio Borges$^1$}
\author{Ram\'on Pino P\'erez$^2$}
\affil{$^1$Departamento de Matemáticas, Universidad Simón Bolívar, Caracas, Venezuela. E-mail: {\tt nborges@usb.ve}\\
$^2$Departamento de Matemáticas, Facultad de Ciencias, Universidad de Los Andes, Mérida, Venezuela. E-mail: {\tt pino@ula.ve}}
\date{\today}

\begin{document}
	
\maketitle

\begin{abstract}
	In  most classical models of belief change, epistemic states are represented by theories (AGM) or formulas (Katsuno-Mendelzon) and the new pieces of information by formulas. The Representation Theorem for revision operators says that operators are represented by total preorders. This important representation is exploited by Darwiche and Pearl to shift the notion of epistemic state to a more abstract one, where the paradigm of epistemic state is  indeed that of a total preorder over interpretations. In this work, we introduce a 3-valued logic where the formulas can be identified with a generalisation of total preorders of three levels: a ranking function mapping interpretations into the truth values. Then we analyse  some sort of changes in this kind of structures and give syntactical characterizations of them.
\end{abstract}

\section{Introduction}

Classical propositional logic is the most common choice when studying and modelling  belief change operators (see for instance \cite{Gar92,Pep08,PU10}).
In the classical AGM framework \cite{AGM85,Gar88}, for instance, epistemic states are represented by theories and new pieces of information by propositional formulas. In the Katsuno-Mendelzon (KM) framework
\cite{KM91}, on the other hand, epistemic states as well as new pieces of information are propositional formulas.
In both AGM and KM settings, the new information is intended to express a fact about the world and this new knowledge must always be included in the epistemic state resulting from the revision process.

A very useful tool in order to understand the logical model is the Representation Theorem for revision operators. It says that operators are represented by assignments mapping epistemic states to total preorders and the output is a formula or theory having as models the minimal models of the new piece of information with respect  to the preorder associated to the old epistemic state.  As a matter of fact,  this tool  is exploited by Darwiche and Pearl \cite{DP97} to shift the notion of epistemic state to a more abstract one, where the paradigm of epistemic state is  indeed that of a total preorder over interpretations.
In our view, this work together with Boutilier's Natural Revision \cite{Bou96} is one of the most influential in a series  
of works trying to capture some kinds of semantical behavior via a syntactical characterization \cite{BM06,BMKR06,JT07,KP08,KMP10,MP13,Nay94,Rot09}.

In this work, we follow this tradition but in a new way. The semantical structures of preorders are represented in a language. In order to do that, we introduce a 3-valued logic with modalities where the formulas can be identified with a generalisation of total preorders of three levels: a ranking function mapping interpretations into the truth values.

The typical situation  we want to model here is as follows.
There is an intelligent agent working within a finite propositional 3-valued logic on $n$ variables. Each variable in this logic, as customary, represents an atomic ``fact'' about the world. We chose a 3-valued logic because this agent does not necessarily have an opinion --or knowledge-- about every single atomic fact.

Our agent has classified all the possible worlds, i.e. all the truth assignments into three blocks $L_1,L_2$ and $L_3$. Those in $L_1$ are the most plausible scenarios. Assignments in $L_2$ are assignments about which the agent is uncertain. She doesn't know wether to accept or reject these worlds. At the level $L_3$ the agent puts the worlds which she considers unlikely.

In the same manner as propositional classical logic captures all the structures at two levels (accepted or rejected), we present here a logic
in which the formulas capture all the three levels' structures. In order to do that, we propose here the use of the Kleene Strong three valued logic \cite{Ber08}
plus two modalities which will be necessary for the completeness of the  representation.

Thus, with the help of this logic, we can model belief change under uncertainty. In order to understand the mechanisms which govern the   changes, we begin with a particular operator we call Cautious Improvement. The way to describe the changes produced in the old epistemic state  by the new piece of information is reminiscent of the changes produced by the improvement operators introduced in \cite{KP08} (see also \cite{KMP10,MP13}).

It is worth to note that in our framework, both the (old) epistemic state and the new piece of information are formulas. Since they are formulas in our new logic, they are both (complex) epistemic states. Therefore, we return, in a natural way, to the paradigm of revising an epistemic state by an epistemic state
first proposed by Benferhat et al. in \cite{BKPP00}.

The  operator of cautious improvement is characterized syntactically. An analysis of the techniques involved in the definition and in the syntactical  characterization allows us to find   a general method for defining all the operators under uncertainty and
extract syntactic postulates that characterize them.

This work is organized as follows. In Section~\ref{the-logic} the Kleene strong 3-valued logic with two new modalities is defined. Therein is proven that every ranking function mapping interpretations into an  ordered scale of three elements can be represented by a formula of this new logic.  Section~\ref{cautious-operator} is devoted to the definition of the {\em cautious improvement operator} and its syntactic characterization.
In Section~\ref{general-characterization} we characterize  each of the $3^9$ (19,683) possible change operators in our logic.
In Section~\ref{related-works} we compare our results with other works in the literature. In Section~\ref{final-remarks}, we conclude with some remarks and give some lines of future development of this work.
Finally, in Appendix~\ref{sec:ProofOfLongTheorem} we give a combinatorial proof of the necessity of our two modalities in order to be able to represent all the ``total preorders'' over interpretations of three levels (Theorem~\ref{the:NotEnoughModality1}).

\section{The Logic}\label{the-logic}

We work within a modal variant of finite $K_3^S$, the Kleene {\em strong} 3-valued logic \cite{Kle38} (see also \cite{Ber08}),
with variables $\seq{x}{n-1}$.
We will call $\Var_n$ the set of these variables.
The usual version of $K_3^S$ has the same syntax as classic propositional logic.
We include the symbol `$\bot$' as the logical constant that always evaluates to $0$.
For the semantics, we have
three possible truth values: $1$ representing truth, $0$ for falsehood and
$1/2$ for {\em non-determined}.
The truth tables for the connectives are the following:

\begin{center}
	\begin{tabular}{c|c|c}
		$P$	&	$Q$	&	$P\land Q$\\
		\hline
		1	&	1	&	1\\
		1	&	1/2	&	1/2\\
		1	&	0	&	0\\
		1/2	&	1	&	1/2\\
		1/2	&	1/2	&	1/2\\
		1/2	&	0	&	0\\
		0	&	1	&	0\\
		0	&	1/2	&	0\\
		0	&	0	&	0
	\end{tabular}
	\qquad
	\begin{tabular}{c|c|c}
		$P$	&	$Q$	&	$P\lor Q$\\
		\hline
		1	&	1	&	1\\
		1	&	1/2	&	1\\
		1	&	0	&	1\\
		1/2	&	1	&	1\\
		1/2	&	1/2	&	1/2\\
		1/2	&	0	&	1/2\\
		0	&	1	&	1\\
		0	&	1/2	&	1/2\\
		0	&	0	&	0
	\end{tabular}
	\qquad
	\begin{tabular}{c|c|c}
		$P$	&	$Q$	&	$P\longrightarrow Q$\\
		\hline
		1	&	1	&	1\\
		1	&	1/2	&	1/2\\
		1	&	0	&	0\\
		1/2	&	1	&	1\\
		1/2	&	1/2	&	1/2\\
		1/2	&	0	&	1/2\\
		0	&	1	&	1\\
		0	&	1/2	&	1\\
		0	&	0	&	1
	\end{tabular}
\\
\mbox{ }
\\
\mbox{ }
\\
\begin{tabular}{c|c}
		$P$	&	$\neg P$\\
		\hline
		1	&	0\\
		1/2	&	1/2\\
		0	&	1	
	\end{tabular}
\end{center}
\bigskip

For the sake of brevity, let's say that a valuation $\omega$ is a {\em quasi-model}
of $\theta$ iff $\omega(\theta)=1/2$ and a {\em countermodel} of $\theta'$ if $\omega(\theta')=0$.
Then we say that $\omega$ is a quasi-model of a set $\Sigma$ iff it is a quasi-model for
all the formulas in $\Sigma$.
The notion of model is the same as in classic propositional logic.
A formula is a {\em contradiction} if it only has countermodels.

Regarding semantics, we have to define some usual symbols:
\begin{enumerate}
	\item For two formulas $\alpha,\beta$ we write $\equiv$ when they have the same truth table.
	\item We write $\alpha\vDash\beta$ when every model of $\alpha$ is a model of $\beta$ i.e.
	we use this symbol with its classic interpretation.
	\item We write $\alpha\myequiv\beta$ to abbreviate $\alpha\vDash\beta$ and $\beta\vDash\alpha$.
\end{enumerate}

Note that in this logic there are no tautologies. Actually, it is easy to check that for all interpretation $\omega$  taking the constant value $\frac 12$, we have that for every formula $\varphi$, $\omega(\varphi)=\frac 12$.

We extend $K_3^S$ by adding the modal operators $\lozenge_1$ and $\square_1$.
Our syntax is also extended by the new formation rule stating that if $\varphi$ is
a formula then $\lozenge_1\varphi$ and $\square_1\varphi$ are also formulas.
The modal operators are interpreted as follows:

\bigskip
\begin{center}
	
\begin{tabular}{c|c}
	$\varphi$	 &	$\lozenge_1 \varphi$\\
	\hline
	1	&	1/2\\
	1/2	&	0\\
	0	&	0i
\end{tabular}
\qquad
\qquad
\begin{tabular}{c|c}
	$\varphi$	 &	$\square_1 \varphi$\\
	\hline
	1	&	1\\
	1/2	&	1\\
	0	&	1/2
\end{tabular}

\end{center}
\bigskip

Intuitively, $\lozenge_1$ worsens the truth value of $\varphi$ and $\square_1$ improves it. It is easily noted that
$\square_1\varphi=\neg\lozenge_1\neg\varphi$, so these are dual operators.

Note that in  $K_3^S$ extended by $\square_1$ there are many tautologies. As a matter of fact, for every formula $\varphi$, the 
formula $ \square_1 \square_1 \varphi$ is a tautology.

With the use of $\lozenge_1$ and $\square_1$ we can find,
given a truth assignment $\omega$,
a formula $\varphi_\omega$ such that its only model is $\omega$.
Indeed, given a truth assignment $\omega=\tup{t_0,t_1,\dots, t_{n-1}}$ on $\Var$,
we have the formula
\begin{equation*}
\varphi_\omega:=	\alpha_0\land \alpha_1\land \dots \land \alpha_{n-1}
\end{equation*}
where each $\alpha_i$ is a formula given by
\begin{equation*}
	\alpha_i=
	\begin{cases}
	x_i	&	\text{if}\quad t_i=1\\
	\neg x_i	&	\text{if}\quad t_i=0\\
	\square_1 x_i\land\square_1 \neg x_i	&	\text{if}\quad t_i=1/2
	\end{cases}
\end{equation*}
This formula evaluates to $1$ if, and only if, all of the $\alpha_i$'s evaluate to $1$,
thus the only model for $\varphi_\omega$ is $\omega$. Notice that it can have more than one quasi-model.
On the other hand, if $\set{\seq{\omega}{k-1}}$ is a set of truth assignments,
the formula
\begin{equation}\label{eq:CaptureOfSetOfModels}
	\varphi_{\seq{\omega}{n-1}}:=\bigvee_{0\leq i\leq n-1}\varphi_{\omega_i}
\end{equation}
has the interpretations in $\set{\seq{\omega}{k-1}}$ as its only models.

We can also give ``normal'' forms that allow us to ``push'' the modal operator $\lozenge_1$
within parentheses. One can easily check that
\begin{equation*}
	\lozenge_1(\theta_1\land\theta_2)\equiv \lozenge_1\theta_1\land\lozenge_1\theta_2
\end{equation*}
and that
\begin{equation*}
	\lozenge_1(\theta_1\lor\theta_2)\equiv \lozenge_1\theta_1\lor\lozenge_1\theta_2
\end{equation*}

On the other hand, for reasons that will become apparent later, we are going to introduce the modal operator
$\lozenge_2$ and its dual $\square_2$ whose semantics are given by the truth tables

\bigskip
\begin{center}
	
	\begin{tabular}{c|c}
		$\varphi$	 &	$\lozenge_2 \varphi$\\
		\hline
		1	&	1\\
		1/2	&	0\\
		0	&	0
	\end{tabular}
	\qquad
	\qquad
	\begin{tabular}{c|c}
		$\varphi$	 &	$\square_2 \varphi$\\
		\hline
		1	&	1\\
		1/2	&	1\\
		0	&	0
	\end{tabular}
	
\end{center}

\bigskip

It is also easy to check that

\begin{equation*}
\lozenge_2(\theta_1\land\theta_2)\equiv \lozenge_2\theta_1\land\lozenge_2\theta_2
\end{equation*}
and
\begin{equation*}
\lozenge_2(\theta_1\lor\theta_2)\equiv \lozenge_2\theta_1\lor\lozenge_2\theta_2
\end{equation*}

We denote by $K_3^S+\lozenge_1+\lozenge_2$ the modal extension of $K_3^S$ by $\lozenge_1$ and $\lozenge_2$
and by $K_3^S+\lozenge_i$ the extension by $\lozenge_i$ for $i=1,2$.
The set of all the formulas in this logic is denoted by $\mathcal F$.

\subsection{Formulas and preorders}

In the  finite case, $K_3^S+\lozenge_1+\lozenge_2$ over the set of
variables
\begin{equation*}
	\Var_n=\set{\seq{x}{n-1}}
\end{equation*}
we have the set
\begin{equation*}
	\mathcal I_n=\set{\seq{\omega}{3^n-1}}
\end{equation*}
containing all the interpretations on $\Var_n$.

Each formula $\varphi$ in $K_3^S+\lozenge_1+\lozenge_2$ induces a partition of
$\mathcal I_n$ into three blocks $L_1(\varphi),L_2(\varphi)$ and $L_3(\varphi)$ defined as
\begin{align*}
	L_1(\varphi)		&	= \setdef{\omega\in\mathcal I_n}{\omega(\varphi)=1}\\
	L_2(\varphi)		&	= \setdef{\omega\in\mathcal I_n}{\omega(\varphi)=1/2}\\
	L_3(\varphi)		&	= \setdef{\omega\in\mathcal I_n}{\omega(\varphi)=0}
\end{align*}
It is easy to check that the relation $\preceq_\varphi$,
defined as
\begin{equation*}
	\omega\preceq_\varphi\omega'	\iff	\omega\in L_i(\varphi), \omega'\in L_j(\varphi)\quad \text{with}\quad i\geq j
\end{equation*}
is a `total preorder'. Thus, the first level contains the worlds accepted by $\varphi$; the second level contains
the uncertain worlds of $\varphi$ and the third level contains the worlds rejected by $\varphi$.

More precisely, we can see $\preceq_\varphi$ as a ranking function $r_\varphi : {\mathcal I_n}\dans \Set{0,\frac 12, 1}$,
where $r_\varphi(\omega)=1$ when $\omega\in L_1(\varphi)$, $r_\varphi(\omega)=\frac 12$ when $\omega\in L_2(\varphi)$ and
$r_\varphi(\omega)=0$ when $\omega\in L_3(\varphi)$.
We regard $L_1(\varphi),L_2(\varphi)$
and $L_3(\varphi)$ as the levels of $\preceq_\varphi$.
Notice that every formula defines one of these three-level preorders (ranking function).
Such a preorder must have at least one non empty level.

Given a ranking function  $r : {\mathcal I_n}\dans \Set{0,\frac 12, 1}$ we also use the notation
 $L_j(P)$ for levels on the preorder $P$ defined by it, that is
 $\omega P \omega'$  iff $r(\omega)\geq r(\omega')$ and $L_1(P)= r^{-1}(1)$, $L_2(P)= r^{-1}(\frac 12)$ and $L_3(P)= r^{-1}(0)$.
Note that these `preorders' are indeed more general than total preorders. For instance, with ranking functions we can distinguish
 different constant functions mapping all the interpretations to different levels, whereas considered as simply total preorders,  all these different ranking functions are identified with the same flat total  preorder. By abuse of notation, in this work we will continue to call three-level preorder $P$, any partition in three levels given by a ranking function $r : {\mathcal I_n}\dans \Set{0,\frac 12, 1}$.

We will prove that for every three-level preorder $P$ on $\mathcal I_n$
there is a formula in $K_3^S+\lozenge_1+\lozenge_2$
that characterizes it.

\begin{theorem}\label{th:SyntacticRepresentation}
	Given a three-level preorder $P$ in $\mathcal I_n$, there is a formula $\phi_P$ in $K_3^S+\lozenge_1+\lozenge_2$ such that
	for every interpretation $\omega$ in $I_n$:
	\begin{enumerate}\denselist
		\item $\omega(\phi_P)=1	\iff\omega\in L_1(P)$,
		\item $\omega(\phi_P)=1/2\iff\omega\in L_2(P)$, and
		\item $\omega(\phi_P)=0\iff \omega\in L_3(P)$.
	\end{enumerate}
\end{theorem}

\begin{proof}
	We know, due to Equation \ref{eq:CaptureOfSetOfModels}, that there are formulas $\psi_1,\psi_2$ and $\psi_3$ such that $\omega(P)=1$ if and only if $\omega\in L_j(P)$. Define the formula
	\begin{equation*}
	\phi_P:=\neg(\lozenge_1\psi_2\lor\lozenge_2\psi_3)
	\end{equation*}
	Then,
	\begin{align*}
	\omega\in L_1(P)	
	&\iff \omega\not\in L_2(P)\quad\text{and}\quad \omega\not\in L_3(P)\\
	&\iff  \omega(\psi_2),\omega(\psi_3)\in\set{0,1/2}\\
	&\iff	\omega(\lozenge_1\psi_2)=\omega(\lozenge_2\psi_3)=0\\
	&\iff	\omega(\lozenge_1\psi_2\lor\lozenge_2\psi_3)=0\\
	&\iff	\omega(\neg(\lozenge_1\psi_2\lor\lozenge_2\psi_3))=1
	\end{align*}
	i.e.,
	\begin{equation*}
	\omega\in L_1(P)\iff \omega(\phi_P)=1
	\end{equation*}

	For the characterization of the second level, we have
	\begin{align*}
	\omega\in L_2(P)	&\implies	\omega(\psi_2)=1\quad\text{and}\quad\omega(\psi_3)\in\set{0,1/2}\\
	&\implies	\omega(\lozenge_1\psi_2)=1/2\quad\text{and}\quad\omega(\lozenge_2\psi_3)=0\\
	&\implies	\omega(\lozenge_1\psi_2\lor\lozenge_2\psi_3)=1/2\\
	&\implies	\omega(\neg(\lozenge_1\psi_2\lor\lozenge_2\psi_3))=1/2
	\end{align*}
	On the other hand,
	\begin{align*}
	\omega(\neg(\lozenge_1\psi_2\lor\lozenge_2\psi_3))=1/2
	&\implies	\omega(\lozenge_1\psi_2\lor\lozenge_2\psi_3)=1/2	\\
	&\implies	\omega(\lozenge_1\psi_2)=1/2\quad\text{or}\quad\omega(\lozenge_2\psi_3)=1/2
	\end{align*}
	It is not possible to have $\omega(\lozenge_2\psi_3)=1/2$ because the truth value of $\lozenge_2\psi_3$
	is either 1 or 0 under any interpretation,
	thus it must be $\omega(\lozenge_1\psi_2)=1/2$, so $\omega(\psi_2)=1$ and as a consequence $\omega\in L_2(P)$.
	
	Hence
	\begin{equation*}
	\omega\in L_2(P)	\iff	\omega(\phi_P)=1/2
	\end{equation*}
	for any interpretation $\omega$.
	
	Finally, for the third level:
	\begin{align*}
	\omega\in L_3(P)			&\implies	\omega(\psi_2)\in\set{0,1/2}\quad\text{and}\quad\omega(\psi_3)=1\\
	&\implies	\omega(\lozenge_1\psi_2)=0\quad\text{and}\quad\omega(\lozenge_2\psi_3)=1\\
	&\implies	\omega(\lozenge_1\psi_2\lor\lozenge_2\psi_3)=1\\
	&\implies	\omega(\neg(\lozenge_1\psi_2\lor\lozenge_2\psi_3))=0
	\end{align*}
	For the converse, we have
	\begin{equation*}
	\omega(\neg(\lozenge_1\psi_2\lor\lozenge_2\psi_3))=0
	\implies
	\omega(\lozenge_1\psi_2\lor\lozenge_2\psi_3)=1
	\end{equation*}
	Thus, either $\omega(\lozenge_1\psi_2)=1$ or $\omega(\lozenge_2\psi_3)=1$.
	Since $\omega(\lozenge_1\theta)\leq 1/2$ for every formula $\theta$,
	it must be the case that $\omega(\lozenge_2\psi_3)=1$,
	so $\omega(\psi_3)=1$ and $\omega\in L_3(P)$.
	Consequently
	\begin{equation*}
	\omega\in L_3(P)	\iff	\omega(\phi_P)=0
	\end{equation*}
	
	Henceforth $\phi_P$ captures the preorder $P$.
\end{proof}

Theorem \ref{th:SyntacticRepresentation} shows that
every preorder on $I_n$ is characterized by a formula in $K_3^S+\lozenge_1+\lozenge_2$.
In Theorem \ref{the:NotEnoughModality1} we prove that both $\lozenge_1$
and $\lozenge_2$ are needed for this.

\begin{theorem}\label{the:NotEnoughModality1}
	$K_3^S+\lozenge_i$ with $i=1,2$ is not enough to define every 3-level preorder on $\mathcal I_n$.
\end{theorem}

\begin{proof}
    See Appendix \ref{sec:ProofOfLongTheorem}.
\end{proof}

\begin{corollary}\label{th:NonCapture1}
	There are three-level preorders on all the interpretations on $\Var_n$ that can not be defined by any formula in $K_3^S$ .
\end{corollary}

\section{The Cautious Improvement Operator}\label{cautious-operator}
We want to define a belief revision operator of the kind 
\[
\ast:\mathcal F\times \mathcal F\longrightarrow \mathcal F
\]
Thus, in this framework  we are going to represent both {\em epistemic states} and {\em epistemic inputs} by formulas in $K_3^S+\lozenge_1+\lozenge_2$.
Also the outputs will be represented by the same type of formulas.

\subsection{Definition of the operator}\label{sec:DefinitionOfCI}

For a fixed $n\in\mathbb N$
we first define $\ast$ as a binary operator on 3-level preorders over $\mathcal I_n$.
The following table gives the level of an interpretation $\omega$ in
$\preceq_\varphi\ast\preceq_\theta$
given the levels where it is located in $\preceq_\varphi$ and $\preceq_\theta$:

\bigskip
\begin{equation}\label{eq:OperatorTable}
\begin{array}{c|c|c|c}
\omega				&	L_1(\preceq_\theta)				&	L_2(\preceq_\theta)					&	L_3(\preceq_\theta)\\
\hline					
L_1(\preceq_\varphi)	&	L_1(\preceq_\varphi\ast\preceq_\theta)	&	L_2(\preceq_\varphi\ast\preceq_\theta)	&	L_2(\preceq_\varphi\ast\preceq_\theta)	\\
\hline
L_2(\preceq_\varphi)	&	L_1(\preceq_\varphi\ast\preceq_\theta)	&	L_2(\preceq_\varphi\ast\preceq_\theta)	&	L_3(\preceq_\varphi\ast\preceq_\theta)	\\
\hline
L_3(\preceq_\varphi)	&	L_2(\preceq_\varphi\ast\preceq_\theta)	&	L_2(\preceq_\varphi\ast\preceq_\theta)	&	L_3(\preceq_\varphi\ast\preceq_\theta)						
\end{array}
\end{equation}
\bigskip

We have defined $\ast$ as a binary operator on preorders.
Thanks to Theorem \ref{th:SyntacticRepresentation}, we can define $\ast$ as a binary operator
on formulas, but still referring to preorders:
given formulas $\varphi$ and $\theta$, the formula $\varphi\ast\theta$
is the one corresponding to the preorder $\preceq_\varphi\ast\preceq_\theta$.
Also notice that $\preceq_\varphi\ast\preceq_\theta\,=\,\preceq_{\varphi\ast\theta}$.

We can make two observations:

First, Table \ref{eq:OperatorTable} yields a truth table for $\varphi\ast\theta$
\begin{equation*}
	\begin{array}{c|c|c}
	\varphi		&	\theta		&		\varphi\ast\theta\\
	\hline
	\hline
	1			&		1		&		1\\
	\hline
	1			&		1/2		&		1/2\\
	\hline	
	1			&		0		&		1/2\\
	\hline
	1/2			&		1		&		1\\
	\hline
	1/2			&		1/2		&		1/2\\
	\hline
	1/2			&		0		&		0\\
	\hline
	0			&		1		&		1/2\\
	\hline
	0			&		1/2		&		1/2\\
	\hline
	0			&		0		&		0
	\end{array}
\end{equation*}

Second, via the identification of formulas with preorders
we can write $L_j(\phi)$ instead of $L_j(\preceq_\phi)$ for any given formula $\phi$.

Figure \ref{fig:Behavior} illustrates the behavior of this operator.

\begin{figure}
\begin{center}	
	\begin{tikzpicture}
	
	\draw[very thick] (-8,-1) -- (-6,-1);
	\draw[very thick] (-8,0) -- (-6,0);
	\draw[very thick] (-8,1) -- (-6,1);
	
	\node at (-7,-1.75) {$\varphi$};
	
	
	\node at (-7.5,-.85) {$\omega_1$};
	\node at (-7,-.85) {$\omega_2$};
	\node at (-6.5,-.85) {$\omega_3$};
	
	\node at (-7.5,.15) {$\omega_4$};
	\node at (-7,.15) {$\omega_5$};
	\node at (-6.5,.15) {$\omega_6$};
	
	\node at (-7.5,1.15) {$\omega_7$};
	\node at (-7,1.15) {$\omega_8$};
	\node at (-6.5,1.15) {$\omega_9$};
	
	
	\node at (-5,0) {$\ast$};
	
	
	\draw[very thick] (-4,-1) -- (-2,-1);
	\draw[very thick] (-4,0) -- (-2,0);
	\draw[very thick] (-4,1) -- (-2,1);
	
	\node at (-3,-1.75)	{$\theta$};
	
	
	\node at (-3.5,1.15) {$\omega_1$};
	\node at (-3.5,.15) {$\omega_2$};
	\node at (-3.5,-.85) {$\omega_3$};
	
	\node at (-3,1.15) {$\omega_4$};
	\node at (-3,.15) {$\omega_5$};
	\node at (-3,-.85) {$\omega_6$};
	
	\node at (-2.5,1.15) {$\omega_7$};
	\node at (-2.5,.15) {$\omega_8$};
	\node at (-2.5,-.85) {$\omega_9$};
	
	
	\node at (-1,0) {$=$};
	
	
	\draw[very thick] (0,-1) -- (2,-1);
	\draw[very thick] (0,0) -- (2,0);
	\draw[very thick] (0,1) -- (2,1);
	
	\node at (1,-1.75)	{$\varphi\ast\theta$};
	
	
	\node at (1.5,1.15) {$\omega_7$};
	\node at (.5,1.15) {$\omega_4$};
	
	\node at (0,.15) {$\omega_1$};
	\node at (0.5,.15) {$\omega_2$};
	\node at (1,.15) {$\omega_5	$};
	\node at (1.5,.15) {$\omega_8$};
	\node at (2,.15) {$\omega_9$};

	\node at (.5,-.85) {$\omega_3$};
	\node at (1.5,-.85) {$\omega_6$};
	\end{tikzpicture}
\end{center}	
	\caption{Example of operator's behavior}\label{fig:Behavior}
\end{figure}

\subsection{Postulates}
	\begin{lemma}\label{le:CICharacterization1}
	The Cautious Improvement Operator satisfies the following postulates	
	\[
	\begin{array}{lr}
	(\varphi\land\theta) \vee ((\square_1\varphi\land\square_1\neg\varphi)\land\theta)		\myequiv \varphi\ast\theta & \mbox{(CI1)}\\
	((\square_1\varphi\land\square_1\neg\varphi)\land\neg\theta)\vee (\neg\varphi\land\neg\theta)		\myequiv \neg(\varphi\ast\theta) &\mbox{(CI2)}
	\end{array}
	\]

	\end{lemma}

\begin{proof}
	\begin{description}
		\item{CI1:} If $\omega((\varphi\land\theta) \lor ((\square_1\varphi\land\square_1\neg\varphi)\land\theta))=1$
		then either $\omega(\varphi\land\theta)=1$ or $\omega((\square_1\varphi\land\square_1\neg\varphi)\land\theta)=1$.
		In the first case, $\omega(\varphi)=\omega(\theta)=1$ thus $\omega\in L_1(\varphi)\cap L_1(\theta)$ and
		this implies by Table \ref{eq:OperatorTable} that $\omega\in L_1(\varphi\ast\theta)$, and so $\omega(\varphi\ast\theta)=1$.
		
		In the second case, we have $\omega(\square_1\varphi\land\square_1\neg\varphi)=\omega(\theta)=1$
		which implies that $\omega\in L_1(\theta)=1$.
		On the other hand,
		\begin{equation*}
			\omega(\square_1\varphi)=1		\implies	\omega(\varphi)\geq 1/2
		\end{equation*}
		and
		\begin{align*}
			\omega(\square_1\neg\varphi)=1		&	\implies	\omega(\neg\varphi)\geq 1/2\\
												&	\implies	\omega(\varphi)\leq 1/2
		\end{align*}
		then
		\begin{equation*}
			1/2\leq \omega(\varphi)\leq 1/2
		\end{equation*}
		thus
		\begin{equation*}
		1/2\leq \omega(\varphi)=1/2
		\end{equation*}
		and $\omega\in  L_2(\varphi)$.
		Hence $\omega\in L_1(\varphi\ast\theta)$ by Table \ref{eq:OperatorTable} and $\omega(\varphi\ast\theta)=1$.
		
		Conversely, if $\omega(\varphi\ast\theta)=1$ then $\omega\in L_1(\varphi\ast\theta)$.
		There are only two cases in Table \ref{eq:OperatorTable} that make this possible: either $\omega\in L_1(\varphi)\cap L_1(\theta)$
		or $\omega\in L_2(\varphi)\cap L_1(\theta)$.
		
		If $\omega\in L_1(\varphi)\cap L_1(\theta)$ then $\omega(\varphi)=\omega(\theta)=1$ and $\omega(\varphi\land\theta)=1$.
		If $\omega\in L_2(\varphi)\cap L_1(\theta)$ then
		$\omega(\square_1\varphi\land \square_1\neg\varphi)=1$ and $\omega(\theta)=1$,
		hence $\omega((\square_1\varphi\land \square_1\neg\varphi)\theta)=1$.
		Thus if $\omega(\varphi\ast\theta)=1$ we have that
		 $\omega((\varphi\land\theta) \lor ((\square_1\varphi\land\square_1\neg\varphi)\land\theta))=1$.

		\item{CI2:} Suppose $\omega$ is an interpretation such that
		$\omega(((\square_1\varphi\land\square_1\neg\varphi)\land\neg\theta)\vee (\neg\varphi\land\neg\theta))=1$.
		Then either $\omega(((\square_1\varphi\land\square_1\neg\varphi)\land\neg\theta))=1$
		or $\omega(\neg\varphi\land\neg\theta)=1$.
		In both cases $\omega(\neg\theta)=1$ which implies that $\omega\in L_3(\theta)$.
		In the first case, $\omega(\square_1\varphi\land\square_1\neg\varphi)=1$ implies that $\omega\in L_2(\varphi)$
		thus $\omega\in L_2(\varphi)\cap L_3(\theta)$.
		In the second case, $\omega(\neg\theta)=1$ implies that $\omega\in L_3(\varphi)\cap L_3(\theta)$.
		In both cases $\omega\in\L_3(\varphi\ast\theta)$ by Table \ref{eq:OperatorTable},
		thus $\omega(\neg(\varphi\ast\theta))=1$.
		
		If $\omega(\neg(\varphi\ast\theta))=1$ then $\omega\in L_3(\varphi\ast\theta)$. The only options
		given in Table \ref{eq:OperatorTable} are $\omega\in L_2(\varphi)\cap L_3(\theta)$
		and $\omega\in L_3(\varphi)\cap L_3(\theta)$.
		\begin{align*}
			\omega\in L_2(\varphi)\cap L_3(\theta)	&	\implies	\omega((\square_1\varphi\land\square_1\neg\varphi)\land\neg\theta)=1\\
			\omega\in L_3(\varphi)\cap L_3(\theta)	&	\implies	\omega(\neg\varphi\land\neg\theta)=1
		\end{align*}
		hence
		\begin{equation*}
			\omega(((\square_1\varphi\land\square_1\neg\varphi)\land\neg\theta)\vee (\neg\varphi\land\neg\theta))=1
		\end{equation*}
	\end{description}
\end{proof}

\begin{lemma}\label{le:CICharacterization2}
	If $\star$ is any binary operator on $\mathcal F$ satisfying CI1 and CI2 then $\star$ is the
	Cautious Improvement operator.
\end{lemma}

\begin{proof}
	Suppose $\star$ satisfies
		\[
	\begin{array}{lr}
	(\varphi\land\theta) \vee ((\square_1\varphi\land\square_1\neg\varphi)\land\theta)		\myequiv \varphi\star\theta & \mbox{(CI1)}\\
	((\square_1\varphi\land\square_1\neg\varphi)\land\neg\theta)\vee (\neg\varphi\land\neg\theta)		\myequiv \neg(\varphi\star\theta) &\mbox{(CI2)}
	\end{array}
	\]
	We have to prove that $L_j(\varphi\star\theta)= L_j(\varphi\ast\theta)$ for $j=1,2,3$,
	i.e., that $\varphi\star\theta\equiv\varphi\ast\theta$.
	
	Given any interpretation $\omega$
	\begin{align*}
		\omega\in L_1(\varphi\star\theta)	&	\iff	
		\omega((\varphi\land\theta) \vee ((\square_1\varphi\land\square_1\neg\varphi)\land\theta))=1\\
											&	\iff	\omega\in L_1(\varphi\ast\theta)
											\qquad\text{(Lemma \ref{le:CICharacterization1})}	
	\end{align*}
	similarly,
	\begin{align*}
	\omega\in L_3(\varphi\star\theta)	&	\iff	
										\omega(((\square_1\varphi\land\square_1\neg\varphi)\land\neg\theta)\vee (\neg\varphi\land\neg\theta))=1\\
										&	\iff	\omega\in L_3(\varphi\ast\theta)
										\qquad\text{(Lemma \ref{le:CICharacterization1})}	
	\end{align*}
	
	Finally, suppose $\omega\in L_2(\varphi\star\theta)$.
	As we are assuming that $\star$ satisfies CI1 and CI2, this is equivalent to
	\begin{align*}
		\omega((\varphi\land\theta) \vee ((\square_1\varphi\land\square_1\neg\varphi)\land\theta))	&	\leq 1/2\\
		\text{and}\\
		\omega(((\square_1\varphi\land\square_1\neg\varphi)\land\neg\theta)\vee (\neg\varphi\land\neg\theta))	&	\leq 1/2
	\end{align*}
	which is equivalent, via Lemma \ref{le:CICharacterization1}, to
	\begin{align*}
	\omega(\varphi\ast\theta)	&	\leq 1/2\\
	\text{and}\\
	\omega(\neg(\varphi\ast\theta))	&	\leq 1/2
	\end{align*}
	This happens if, and only if
	\begin{align*}
	\omega(\varphi\ast\theta)	&	\leq 1/2\\
	\text{and}\\
	\omega(\varphi\ast\theta)	&	\geq 1/2
	\end{align*}
	i.e.
	\begin{equation*}
		1/2\leq \omega(\varphi\ast\theta)\leq 1/2
	\end{equation*}
	or, equivalently, $\omega(\varphi\ast\theta)=1/2$
	which is equivalent to $\omega\in L_2(\varphi\ast\theta)$.
	
	Hence $\omega\in L_2(\varphi\star\theta)$ iff $\omega\in L_2(\varphi\ast\theta)$
	thus $L_2(\varphi\star\theta)=L_2(\varphi\ast\theta)$.
	
	Henceforth, $\varphi\star\theta\equiv\varphi\ast\theta$.
\end{proof}

\begin{theorem}\label{th:CharacterizationOfCI}
	Postulates CI1 and CI2 characterize the Cautious Improvement operator.
\end{theorem}

\begin{proof}
	Lemmas \ref{le:CICharacterization1} and \ref{le:CICharacterization2}.
\end{proof}

We can obtain simpler, yet less intuitive, expressions for CI1 and CI2.

\begin{theorem}\label{th:Equivs}
	The following are equivalent, respectively, to CI1 and CI2:
	\begin{align}
	\varphi\ast\theta	&	\equiv	\square_1\varphi\land\theta \label{eq:Equiv1} \tag{CI1'} \\
	\neg(\varphi\ast\theta)	&	\equiv	\square_1\neg\varphi\land\neg\theta \label{eq:Equiv3} \tag{CI2'}
	\end{align}
\end{theorem}	

\begin{proof}
	According to CI1,
	\begin{equation*}
	(\varphi\land\theta)\lor (\square_1\varphi\land\square_1\neg\varphi\land\theta)\myequiv \varphi\ast \theta
	\end{equation*}
	It is easy to  check, using truth tables, that	
	\begin{equation*}
	(\varphi\land\theta)	\lor	(\square_1\varphi\land \square_1\neg\varphi\land\theta)
		\equiv	(\varphi	\lor	(\square_1\varphi\land \square_1\neg\varphi))\land \theta
	\end{equation*}
	and
	\begin{equation*}
	(\varphi	\lor	(\square_1\varphi\land \square_1\neg\varphi))\equiv \square_1\varphi
	\end{equation*}
	hence
	\begin{equation*}
	\varphi\ast\theta\myequiv \square_1\varphi\land\theta
	\end{equation*}
	
	On the other hand, by CI2
	\begin{equation*}
		((\square_1\varphi\land\square_1\neg\varphi)\land\neg\theta)\vee (\neg\varphi\land\neg\theta)		\myequiv \neg(\varphi\ast\theta)
	\end{equation*}
	And again, using truth tables, it is easy to check that
		\begin{equation*}
	(\square_1\varphi\land \square_1\neg\varphi\land\neg\theta)\lor (\neg\varphi\land\neg\theta)
	\equiv	((\square_1\varphi\land \square_1\neg\varphi)\lor\neg\varphi)\land \neg\theta
	\end{equation*}
	and
	\begin{equation*}
	\varphi	\lor	(\square_1\varphi\land \square_1\neg\varphi)\equiv \square_1\neg\varphi
	\end{equation*}
	thus
	\begin{equation*}
	\neg(\varphi\ast\theta)\myequiv \square_1\neg\varphi\land\neg\theta
	\end{equation*}
		
\end{proof}

It is worth to note that Theorem~\ref{th:Equivs} tells us that the cautious improvement operator can be expressed as a formula of of the logic $K^S_3+\lozenge_1 +\lozenge_2$ in terms of their inputs.

\begin{theorem}
	Given any two formulas $\varphi$ and $\theta$:

\begin{align*}
\neg(\varphi\ast\theta)	&	\equiv \neg\varphi\ast\neg\theta \tag{CI3}\\
\text{If $\theta$ is not a contradiction, } & \text{then $\varphi\ast\theta$ is not a contradiction} \label{rule:nocontradiction}\tag{CI4}\\
\varphi\ast\theta	&	\models\theta  \label{rule:Bounding}\tag{CI5}\\
\varphi	&	\models\square_1(\varphi\ast\theta)  \label{rule:NoCountermodels}\tag{CI6}\\
(\varphi\ast\theta)\ast\theta	&	\equiv\theta \label{rule:iteration}\tag{CI7}\\
\theta\ast\theta	&	\equiv \theta  \label{rule:idempotency}\tag{CI8}
\end{align*}

\end{theorem}

\begin{proof}
	These six properties can be checked out using truth tables.
\end{proof}

Properties \ref{rule:nocontradiction} and \ref{rule:Bounding} respectively
represent the Non Contradiction and the Success principles.
On the other hand, \ref{rule:iteration} and \ref{rule:idempotency} tell us how
 the iterative behavior of this operator works.
It takes exactly two steps to completely replace the original epistemic state
with the new information.

\section{General definition of a belief change operator}\label{general-characterization}

Now we want to define a general belief change operator.
\[
\otimes:\mathcal F\times \mathcal F\longrightarrow \mathcal F
\]

A function
\begin{equation*}
k:\set{1,2,3}^2\longrightarrow \set{1,2,3}
\end{equation*}
defines a binary operator $\otimes$ on three-level preorders over $\mathcal I_n$
given by the table:

\bigskip
\begin{equation}\label{eq:GeneralOperatorTable}
\begin{array}{c|c|c|c}
\omega				&	L_1(\preceq_\theta)				&	L_2(\preceq_\theta)					&	L_3(\preceq_\theta)\\
\hline					
L_1(\preceq_\varphi)	&	L_{k(1,1)}(\preceq_\varphi\otimes\preceq_\theta)	&	L_{k(1,2)}(\preceq_\varphi\otimes\preceq_\theta)	&	L_{k(1,3)}(\preceq_\varphi\otimes\preceq_\theta)	\\
\hline
L_2(\preceq_\varphi)	&	L_{k(2,1)}(\preceq_\varphi\otimes\preceq_\theta)	&	L_{k(2,2)}(\preceq_\varphi\otimes\preceq_\theta)	&	L_{k(2,3)}(\preceq_\varphi\otimes\preceq_\theta)	\\
\hline
L_3(\preceq_\varphi)	&	L_{k(3,1)}(\preceq_\varphi\otimes\preceq_\theta)	&	L_{k(3,2)}(\preceq_\varphi\otimes\preceq_\theta)	&	L_{k(3,3)}(\preceq_\varphi\otimes\preceq_\theta)						
\end{array}
\end{equation}
\bigskip

Given formulas $\varphi$ and $\theta$,
an interpretation $\omega$ is in $L_{k(i,j)}(\preceq_\varphi\otimes\preceq_\theta)$
if and only if $\omega\in L_i(\varphi)\cap L_j(\theta)$.
This defines the preorder $\preceq_\varphi\otimes\preceq_\theta$.

Since the function $k$ determines the operator $\otimes$,
there are $|\set{1,2,3}|^{|\set{1,2,3}|^2}=3^{3^2}=3^9$ possible belief change operators in this setting.
Different functions yield operators with different epistemological attitudes.
The Cautious Improvement operator defined in Section \ref{sec:DefinitionOfCI}, for instance,
gives priority to new information (i.e. to the models of $\theta$) but
does so in a {\em cautious} manner.
For example, if an interpretation $\omega$ is a counter model of $\varphi$ and it is
a model of $\theta$ then it will be, not a model, but a quasi model of $\varphi\ast\theta$.
Another example could be the one given by the following table:
\bigskip
\begin{equation*}\label{eq:DrasticOperatorTable}
\begin{array}{c|c|c|c}
\omega				&	L_1(\preceq_\theta)				&	L_2(\preceq_\theta)					&	L_3(\preceq_\theta)\\
\hline					
L_1(\preceq_\varphi)	&	L_1(\preceq_\varphi\circ\preceq_\theta)	&	L_2(\preceq_\varphi\circ\preceq_\theta)	&	L_3(\preceq_\varphi\circ\preceq_\theta)	\\
\hline
L_2(\preceq_\varphi)	&	L_1(\preceq_\varphi\circ\preceq_\theta)	&	L_2(\preceq_\varphi\circ\preceq_\theta)	&	L_3(\preceq_\varphi\circ\preceq_\theta)	\\
\hline
L_3(\preceq_\varphi)	&	L_1(\preceq_\varphi\circ\preceq_\theta)	&	L_2(\preceq_\varphi\circ\preceq_\theta)	&	L_3(\preceq_\varphi\circ\preceq_\theta)						
\end{array}
\end{equation*}
\bigskip

It is easy to see that $\varphi\circ\theta\equiv\theta$ for every pair of formulas $\varphi$ and $\theta$.
This operator gives absolute priority to the new information.
Also notice that
\begin{equation*}
	\varphi\circ\theta\equiv \varphi\ast^2\theta\equiv\theta
\end{equation*}

Once defined $\otimes$ as a binary operator on preorders, we can use
Theorem \ref{th:SyntacticRepresentation} to define $\otimes$ as a binary operator
on formulas.

From Table \ref{eq:GeneralOperatorTable} we can obtain a syntactic characterization of $\otimes$.
In order to do this, we need to define a set of formulas
\begin{equation*}
	\setdef{\zeta_{ij}^k}{i,j,k\in\set{1,2,3}}
\end{equation*}
as follows.
For every triple $(i,j,k)$:
\begin{enumerate}
	\item If $k\neq k(i,j)$ then $\zeta_{ij}^k=\bot$.
	
	\item If $k=k(i,j)$ then
		\begin{equation*}
		\zeta_{ij}^k=\alpha_{i}^k\land \beta_{j}^k
		\end{equation*}
		where
		\begin{equation*}
		\alpha_{i}^k=
		\begin{cases}
		\varphi											&	\text{ if }	i=1\\
		\square_1\varphi\land	\square_1\neg\varphi	&	\text{ if }	i=2\\
		\neg\varphi										&	\text{ if }	i=3
		\end{cases}
		\end{equation*}
		and
		\begin{equation*}
		\beta_{j}^k=
		\begin{cases}
		\theta											&	\text{ if }	j=1\\
		\square_1\theta\land	\square_1\neg\theta	&	\text{ if }	j=2\\
		\neg\theta										&	\text{ if }	j=3
		\end{cases}
		\end{equation*}
		\end{enumerate}

Observe that for any interpretation $\omega$, regardless the value of $k$,
$\omega(\alpha^k_i)=1$ iff $\omega\in L_i(\varphi)$ and
$\omega(\beta^k_j)=1$ iff $\omega\in L_j(\varphi)$.

Now we are able to give a syntactic definition of $\otimes$.
\begin{theorem}\label{th:GeneralPostulates}
The following postulates characterize the operator $\otimes$:
\begin{align*}
\bigvee_{1\leq i\leq j\leq 3} \zeta_{ij}^1 & \myequiv (\varphi\ast\theta)	\tag{$\otimes$1}\\
\bigvee_{1\leq i\leq j\leq 3} \zeta_{ij}^2 & \myequiv \square_1(\varphi\ast\theta)\land \square_1\neg(\varphi\ast\theta)	\tag{$\otimes$2}\\
\bigvee_{1\leq i\leq j\leq 3} \zeta_{ij}^3 & \myequiv \neg(\varphi\ast\theta)	\tag{$\otimes$3}
\end{align*}	
\end{theorem}

We divide the proof of Theorem \ref{th:GeneralPostulates} into two lemmas:

\begin{lemma}
	The operator $\otimes$ defined by the Table \ref{eq:GeneralOperatorTable} satisfies
	Postulates $\otimes1$, $\otimes2$ and $\otimes3$.
\end{lemma}

\begin{proof}
	Let us denote:
	\begin{align*}
		\phi_1	&	=	\varphi\ast\theta\\
		\phi_2	&	=	\square_1(\varphi\ast\theta)\land \square_1\neg(\varphi\ast\theta)\\
		\phi_3	&	=	\neg(\varphi\ast\theta)
	\end{align*}
	Notice that for any interpretation $\omega$
	\begin{equation*}
		\omega(\phi_k)=1 \iff \omega \in L_k(\varphi\ast\theta)
	\end{equation*}

		Fix $k\in\set{1,2,3}$. Suppose $\omega$ is an interpretation that satisfies $\bigvee_{1\leq i\leq j\leq 3} \zeta_{ij}^k$.
		Thus, there is a pair $i,j$ such that $\zeta_{ij}^k\neq \bot$, i.e., $\zeta_{ij}^k$
		is a conjunction $\alpha_{i}^k\land \beta_{j}^k$ such that
		\begin{equation*}
			\omega(\alpha_{i}^k\land \beta_{j}^k)=1.
		\end{equation*}
		That is, $\omega(\alpha_{i}^k)=\omega(\beta_{j}^k)=1$, hence $\omega\in L_i(\varphi)\cap L_j(\theta)$ and 
		thus $\omega\in L_{k(i,j)}(\preceq_\varphi\otimes\preceq_\theta)$ by Table \ref{eq:GeneralOperatorTable}
		and $\omega(\phi_k)=1$.
		
		Conversely, suppose that $\omega(\phi_k)=1$, then $\omega \in L_k(\varphi\ast\theta)$.
		According to the definition of $\otimes$ by Table \ref{eq:GeneralOperatorTable} there is a pair $i,j$
		such that $k=k(i,j)$ thus $\omega\in L_i(\varphi)\cap L_j(\theta)$.
		This implies that $\omega(\alpha_{i}^k\land \beta_{j}^k)=1$, hence
		\begin{equation*}
			\omega\left(\bigvee_{1\leq i\leq j\leq 3} \zeta_{ij}^k\right)=1
		\end{equation*}
\end{proof}

\begin{lemma}
If an operator $\star$ satisfies Postulates $\otimes1$, $\otimes2$ and $\otimes3$ then $\star$ and $\otimes$
are the same operator.
\end{lemma}

\begin{proof}
	Analogous to the proof of Lemma \ref{le:CICharacterization2}.
\end{proof}

%

\section{Related works}\label{related-works}
There are some interesting works linking modal logic and belief revision. Namely, the works of Giacomo Bonanno \cite{Bon05,Bon07}.
The aim of these works is to give modal logics in which the process of revision can be simulated and the original postulates of the AGM framework
are satisfied. Our aim is different in at least two aspects. First, the processes of change we want to capture in our model differ from
the changes proposed in the AGM or KM (Katsuno-Mendelzon) framework where the believes before the change (and after it) are represented either as a classical propositional theory or as  a classical propositional formula. We model change in a bit more complex structures. The purpose of the 3-valued logic with modalities  we use is to be able to represent these structures through the formulas of the logic and the semantics associated.
The second aspect in which our approach differs from those works  is in  our postulates, which come naturally from the complex process we want to model.
We may say that our approach to the syntactic postulates is  oriented by semantics and this leads to the discovery of the axioms.

It is worth to note that our modalities of type one ($\lozenge_1$ and $\square_1$) are related to ideas on improvement operators
\cite{KP08,KMP10,MP13}. Actually, $\square_1$ is similar to the operator of one improvement introduced in those works: the effect of this modality consists of improving by one degree  the plausibility of all worlds, whenever possible.  Of course, our structures of three levels do not allow to simulate completely
operators as the one-improvement operator, because  these operators can create new levels in the process of revision. In order to simulate that, we should have structures with at least $3^n$ levels, the maximal length of a total preorder on $\mathcal{I}_n$.

Note that Hans Root in his work {\em Shifting Priorities: Simple Representations for Twenty-Seven Iterated Theory Change Operators} \cite{Rot09}
analyses the behavior of 27 change operators looking at the changes in spheres' systems (alias total preorders). Our work is reminiscent of this approach 
and the subtitle of our work is a sort of tribute to his work.

\section{Final remarks and future research}\label{final-remarks}

\subsection{Remarks}

This work is a first step in order to have a logic in which formulas represent complex epistemic states, such as ranking functions on interpretations, a generalization of total preorders. We have concentrated here in ranking functions taking three values: acceptation, rejection and indetermination.
Mainly, there are three contributions in this work which we want to remark here:

\begin{enumerate}
\item We have defined a modal expansion of the Kleene's strong 3-valued logic
that allows us to describe three-level preorders associated with
any formula $\varphi$ in this logic. Level $L_1(\varphi), L_2(\varphi)$
and $L_3(\varphi)$ are respectively the sets of models, quasi models and countermodels of $\varphi$.

\item We have defined semantically all the change operators ($3^9$) and, what is most important, we have introduced a technique to characterize them syntactically in the  Kleene's strong 3-valued logic with modalities.
    
\item We have concentrated in a particular, natural  and meaningful operator, called cautious improvement operator. This operator has been characterized through two postulates. Moreover, there is a formula of the new logic constructed in terms of two formulas which are the input of the operator (the old  beliefs and the new piece of information) such that this formula captures the output of the operator (Theorem~\ref{th:Equivs}).

\item We have characterized all possible change operators (19,683 in total) in the three-level structures following the methodology used for characterizing the cautious improvement operator.
\end{enumerate}

\subsection{Future work}
The main issues we want to develop next are the following:
\begin{itemize}
\item To define a complete proof theory  for Kleene's strong 3-valued logic expanded with $\lozenge_1$ and $\lozenge_2$.

\item To characterize a class of operators in which the new piece of information does not worsen, that is, it shifts downwards.

\item To understand when the operators have a representation as in Theorem~\ref{th:Equivs}, that is, when 
the formula corresponding to $\varphi\ast\theta$ can  be expressed
in terms of $\varphi$ and $\theta$ within $K_3^S+\lozenge_1+\lozenge_2$.

\item To generalize the Kleene's strong 3-valued logic with modalities to a modal  $3^n$-valued logic  which can capture all the ranking functions into $3^n$ values and find the way to encode in such a logic all known classes of operators.

\end{itemize}


\appendix

\section{Proof of Theorem \ref{the:NotEnoughModality1}}\label{sec:ProofOfLongTheorem}

We consider $K_3^S+\lozenge_i$ for $i=1,2$ with only one variable.
	We have the interpretations $\omega_1,\omega_2,\omega_3$
	such that
	$\omega_1(x)=1$, $\omega_2(x)=1/2$
	and $\omega_3(x)=0$.
	Consider the set $\mathcal P=\set{\seq{P}{26}}$ of all  preorders on
	these three interpretations, being $P_0$ the preorder
	\begin{center}
		\begin{tabular}{c}
			$\omega_3$\\
			$\omega_2$\\
			$\omega_1$
		\end{tabular}
	\end{center}
	This is represented by the formula $\varphi_0:=x$ in the sense that
	$\omega_1(\varphi_0)=1$, $\omega_2(\varphi_0)=1/2$ and $\omega_3(\varphi_0)=0$.

	\begin{enumerate}
		\item In order to prove this result for $K_3^S+\lozenge_1$,
		we define the operations $\neg, \square_1$ and $\lor$ on the set of preorders.
		Given preorders $P,P'$:
		\begin{enumerate}
			\item $L_1(\neg P)=L_3(P)$, $L_2(\neg P)=L_2(P)$ and $L_3(\neg P)=L_1(P)$.
			\item $L_1(\square_1P)=L_1(P)\cup L_2(P)$, $L_2(\square_1 P)=L_3(P)$ and $L_3(\square_1 P)=\emptyset$.
			\item If interpretation $\omega$
			is at level $L_i(P)$ and level $L_j(P')$,
			then $\omega$ is at level $L_k(P\lor P')$
			where $k=\min\set{i,j}$.
			
		\end{enumerate}
		Let us call $\overline{P_0}^1$ the closure of the set $\set{P_0}$ under these three operations
		and $\overline{P_0}_k^1$ the set of the preorders obtained after
		$k$ or less successive applications of $\set{\neg,\square_1,\lor}$.
		
		Let us say the preorder $P$ belongs to the set $\mathcal F_1$ if and only if it has one of the
		following forms:
		\begin{gather}
		\begin{array}{c}
		\omega\\
		\omega'\\
		\omega_2
		\end{array}
		\qquad
		\begin{array}{c}
		\omega_2\\
		\omega'\\
		\omega
		\end{array}
		\qquad
		\begin{array}{c}
		\omega\\
		\emptyset\\
		\omega'\omega''
		\end{array}
		\qquad
		\begin{array}{c}
		\omega'\omega''\\
		\emptyset\\
		\omega
		\end{array}
		\end{gather}

		We will prove that no preorder in $\mathcal F_1$ belongs to $\overline{P_0}^1$.
		We do it by induction in the number of steps, being our base case
		the application of the operations to $P_0$.
		These yield:
		\begin{equation}\label{eq:BaseCaseNonSufficient1}
		\neg P_0=
		\begin{tabular}{c}
		$\omega_1$\\
		$\omega_2$\\
		$\omega_3$
		\end{tabular}
		\qquad
		\square_1 P_0=
		\begin{tabular}{c}
		$\emptyset$\\
		$\omega_3$\\
		$\omega_1\omega_2$
		\end{tabular}
		\qquad
		P_0\lor P_0=P_0
		\end{equation}
		Obviously none of those preorders belongs to $\mathcal F_1$.
		
		Now suppose that no preorder in $\mathcal F_1$ belongs to $\overline{P_0}_k^1$.
		
		Given a sequence $P_{i_0},P_{i_1},\dots, P_{i_k},P_{i_{k+1}}$ of preorders with $i_0=0$
		such that $P_{i_{j+1}}$ is obtained from $P_{i_j}$ by an application of one of our operators,
		suppose, for a contradiction, that $P_{i_{k+1}}$ is in $\mathcal F_1$.
		Then:
		\begin{enumerate}
			\item If $P_{i_{k+1}}=\neg P_{i_{k}}$ then $P_{i_{k}}$ must be in $\mathcal F_1$ which is a contradiction.
			
			\item If $P_{i_{k+1}}=\square_1 P_{i_{k}}$ then $L_3(P_{i_{k+1}})$ is empty
			thus it is not a linear order and it has a level which is empty other than $L_2(P_{i_{k+1}})$
			hence $P_{i_{k+1}}$ is not in $\mathcal F_1$.

			\item If $P_{i_{k+1}}=P_{i_{k}}\lor P'$ with $P'$ a preorder in $\overline{P_0}_k^1$:
			\begin{itemize}
				\item Suppose $P_{i_{k+1}}$ is a linear order with $\omega_2$ in $L_3$, i.e.
				
				\begin{center}
					$P_{i_{k+1}}$=
					\begin{tabular}{c}
						$\omega_2$\\
						$\omega$\\
						$\omega'$
					\end{tabular}
				\end{center}
				then $\omega_2$ must be in both $L_3(P_{i_{k}})$ and $L_3(P')$ and $\omega'$ must be in $L_1(P_{i_{k}})$
				or $L_1(P')$.
				Without loss of generality, we may assume that $\omega_2\in L_3(P')$ and $\omega'\in L_1(P')$.
				Then necessarily either $\omega\in L_2(P')$ and $P'=P_{i_{k+1}}$
				or $L_2(P')=\emptyset$ in $P'$.
				In both cases, we can conclude that $P'\in\mathcal{F}_1$.
				
				\item Similarly, if $P_{i_{k+1}}$ is a linear order with $\omega_2$ in $L_1$
				
				\begin{center}
					$P_{i_{k+1}}$=
					\begin{tabular}{c}
						$\omega$\\
						$\omega'$\\
						$\omega_2$
					\end{tabular}
				\end{center}
				then $\omega$ must be in both $L_3(P_{i_{k}})$ and $L_3(P')$ and $\omega_2$ must be either in $L_1(P_{i_{k}})$
				or $L_1(P')$. Again, suppose without loss of generality that $P'$ has $\omega$ in $L_3$ and $\omega_2$ in $L_1$.
				If $\omega'$ is not in $L_2$ for $P'$ then this level is empty and if $\omega'$ is in $L_2$ then
				$P'=P_{i_{k+1}}$. In both cases, we have $P'\in\mathcal F_1$.
				
				\item If $P_{i_{k+1}}$ has the form
				
				\begin{center}
					\begin{tabular}{c}
						$\omega$\\
						$\emptyset$\\
						$\omega'$ $\omega''$
					\end{tabular}
				\end{center}
				then necessarily:
				\begin{itemize}
					\item $\omega$ belongs to level $L_3(P_k)$ and $L_3(P')$.
					\item $\omega'$ must be either in $L_1(P_k)$ or $L_1(P')$.
					\item $\omega$ must be either in $L_1(P_k)$ or $L_1(P')$.
				\end{itemize}
				
				Suppose that $\omega'$ and $\omega''$ are in $L_1(P')$, then
				
				\begin{center}
					$P'=$
					\begin{tabular}{c}
						$\omega$\\
						$\emptyset$\\
						$\omega'$ $\omega''$
					\end{tabular}
				\end{center}
				thus $P'\in\mathcal F_1$.
				If this is not the case then necessarily
				
				\begin{center}
					$P'=$
					\begin{tabular}{c}
						$\omega$\\
						$\omega'$\\
						$\omega''$
					\end{tabular}
					\quad
					or
					\quad
					$P'=$
					\begin{tabular}{c}
						$\omega$\\
						$\omega''$\\
						$\omega'$
					\end{tabular}
				\end{center}
				We examine only the case
				
				\begin{center}
					$P'=$
					\begin{tabular}{c}
						$\omega$\\
						$\omega'$\\
						$\omega''$
					\end{tabular}
				\end{center}
				since the other one is symmetric.
				In this case, necessarily
				\begin{center}
					$P_k=$
					\begin{tabular}{c}
						$\omega$\\
						$\omega''$\\
						$\omega'$
					\end{tabular}
				\end{center}
				then one of those two preorders is linear with $\omega_2$ in one of the extreme levels, hence one of them
				is in $\mathcal{F}_1$
				
				\item Finally, suppose that If $P_{i_{k+1}}$ has the form
				
				\begin{center}
					\begin{tabular}{c}
						$\omega'$ $\omega''$\\
						$\emptyset$\\
						$\omega$
					\end{tabular}
				\end{center}
				
				then both $P_k$ and $P'$ must have $\omega'$ and $\omega''$ in $L_3$
				and one of them has to have $\omega$ in $L_1$. Suppose that $P'$ satisfies these two properties,
				then $P'=P_{i_{k+1}}$ and as a consequence $P'\in\mathcal F_1$.
				
			\end{itemize}
			As we can see, on each case we contradict the inductive hypothesis, thus it is not possible that
			$P_{i_{k+1}}\in\mathcal F_1$ without incurring in a contradiction. Hence, $P_{i_{k+1}}\notin\mathcal F_1$
			thus no element of $\mathcal F_1$ belongs to the closure of $P_0$ under the given operations.
			
			As a consequence of this, there are preorders in $\mathcal P$ that have no corresponding formula
			in $K_3^S+\lozenge_1$.

		\end{enumerate}
		
		\item Now to prove that $K_3^S+\lozenge_2$ , in the finite case, cannot capture all the preorders,
		we define a new operator $\square_2$ for every preorder $P$ as
		\begin{align}
		L_1(\square_2 P)	&	=	L_1(P) \cup L_2(P)\\
		L_2(\square_2 P)	&	=	\emptyset\\
		L_3(\square_2 P)	&	=	L_3(P)
		\end{align}

		Let us denote by $\overline{P_0}^2$ the closure of $\set{P_0}$ under operators $\neg, \square_2, \lor$
		and by  $\overline{P_0}^2_k$ the set of the preorders obtained by $k$ or less successive applications of these operators.
		Let $\mathcal F_2$ denote the set of the preorders with one of
		the following configurations:
		\begin{center}
			\begin{tabular}{c}
				$\omega_2$\\
				$\omega'$\\
				$\omega$
			\end{tabular}
			\quad
			\begin{tabular}{c}
				$\omega$\\
				$\omega'$\\
				$\omega_2$
			\end{tabular}
			\quad
			\begin{tabular}{c}
				$\emptyset$\\
				$\omega\omega'\omega''$\\
				$\emptyset$
			\end{tabular}
			\quad
			\begin{tabular}{c}
				$\emptyset$\\
				$\omega$\\
				$\omega_2\omega'$
			\end{tabular}
			\quad
			\begin{tabular}{c}
				$\omega_2\omega'$\\
				$\omega$\\
				$\emptyset$
			\end{tabular}
			\quad
			\begin{tabular}{c}
				$\emptyset$\\
				$\omega\omega'$\\
				$\omega''$
			\end{tabular}
			\quad
			\begin{tabular}{c}
				$\omega''$\\
				$\omega\omega'$\\
				$\emptyset$
			\end{tabular}
		\end{center}

		We will prove, using induction in the number of steps, that $P\in\overline{P_0}^2$
		implies that $P$ is not in $\mathcal F_2$.

		Our base case consists of applying this operators to $P_0$.
		As we saw in Equation \ref{eq:BaseCaseNonSufficient1},
		neither $\neg P_0$ nor $P_0\lor P_0$ belong to $\mathcal F_2$.
		On the other hand
		\begin{equation*}
		\square_2 P_0=
		\begin{tabular}{c}
		$\omega_3$\\
		$\emptyset$\\
		$\omega_1\omega_2$
		\end{tabular}
		\end{equation*}
		which is clearly not in $\mathcal F_2$.
		
		For the inductive hypothesis, assume that no preorder in $\overline{P_0}^2_k$ belongs to $\mathcal F_2$.
		
		Suppose we have a sequence $P_{i_0},P_{i_1},\dots, P_{i_k},P_{i_{k+1}}$ of preorders with $i_0=0$
		such that $P_{i_{j+1}}$ is obtained from $P_{i_j}$ by an application of one of our operators
		and assume for a contradiction that $P_{i_{k+1}}$ is in $\mathcal F_2$. Hence:
		
		\begin{enumerate}
			\item If $P_{i_{k+1}}=\neg P_{i_{k}}$ it is immediate that $P_{i_{k+1}}\in\mathcal F_2$
			iff $P_{i_{k}}\in\mathcal F_2$.
			
			\item Suppose $P_{i_{k+1}}=\square_2 P_{i_{k}}$, then $L_2(P_{i_{k+1}})$ is empty, thus
			$P_{i_{k+1}}$ does not belong to $\mathcal F_2$.
			
			\item If $P_{i_{k+1}}=P_{i_{k}}\lor P'$ with $P'$ a preorder in $\overline{P_0}_k^2$ we need to consider
			seven cases.
			
			\begin{enumerate}
				\item If $P_{i_{k+1}}$ has the form
				
				\begin{center}
					\begin{tabular}{c}
						$\omega_2$\\
						$\omega'$\\
						$\omega$
					\end{tabular}
				\end{center}
				then $\omega_2$ belongs to $L_3(P_{i_{k}})$ and $L_3(P')$ and $\omega$ must be either in
				$L_1(P_{i_{k}})$ or $L_1(P')$. If $\omega$ is in $L_1(P')$ then $P'$ must have the form
				\begin{center}
					\begin{tabular}{c}
						$\omega_2\omega'$\\
						$\emptyset$\\
						$\omega$
					\end{tabular}
				\end{center}
				since otherwise it would be in $\mathcal F_2$. But then $\omega'$ must be in $L_2(P_{i_{k}})$ and
				as a consequence we obtain that $P_{i_{k}}$ is in $\mathcal F_2$ wherever we place $\omega$
				because $P_{i_{k}}$ would be in one of the following forms
				\begin{center}
					\begin{tabular}{c}
						$\omega_2$\\
						$\omega'$\\
						$\omega$
					\end{tabular}
					\qquad
					\begin{tabular}{c}
						$\omega_2$\\
						$\omega'\omega$\\
						$\emptyset$
					\end{tabular}
					\qquad
					\begin{tabular}{c}
						$\omega_2\omega$\\
						$\omega'$\\
						$\emptyset$
					\end{tabular}
				\end{center}
				
				\item If $P_{i_{k+1}}$ is in the form
				\begin{center}
					\begin{tabular}{c}
						$\omega$\\
						$\omega'$\\
						$\omega_2$
					\end{tabular}
				\end{center}
				then $\omega\in L_3(P_{i_k})$ and $\omega\in L_3(P')$
				and $\omega_2\in L_1(P_{i_k})$ or $\omega_2\in L_1(P')$.
				Suppose $\omega\in L_3(P')$ and $\omega_2\in L_1(P')$,
				then we must have $\omega'\in L_3(P')$ or, otherwise,
				$P'$ belongs to $\mathcal F_2$ thus $P'$ must have the form
				\begin{center}
					\begin{tabular}{c}
						$\omega\omega'$\\
						$\emptyset$\\
						$\omega_2$
					\end{tabular}
				\end{center}
				But in this case we must have $\omega'\in L_2(P_{i_k})$
				and this implies that $P_{i_k}$
				is in one of the following forms:
				\begin{center}
					\begin{tabular}{c}
						$\omega$\\
						$\omega'$\\
						$\omega_2$
					\end{tabular}
					\qquad
					\begin{tabular}{c}
						$\omega\omega_2$\\
						$\omega'$\\
						$\emptyset$
					\end{tabular}
					\qquad
					\begin{tabular}{c}
						$\omega$\\
						$\omega_2\omega'$\\
						$\emptyset$
					\end{tabular}
				\end{center}
				hence $P'\notin\mathcal F_2\implies P_{i_k}\in\mathcal F_2$.
				
				\item Suppose $P_{i_{k+1}}$ has the form
				
				\begin{center}
					\begin{tabular}{c}
						$\emptyset$\\
						$\omega\omega'\omega''$\\
						$\emptyset$
					\end{tabular}
				\end{center}
				If neither $P_{i_{k}}$ nor $P'$ have this form, there are two possible combinations that yield this:
				\begin{center}
					\begin{tabular}{c}
						$\omega$\\
						$\omega'\omega''$\\
						$\emptyset$
					\end{tabular}
					\quad
					$\lor$
					\quad
					\begin{tabular}{c}
						$\omega'$\\
						$\omega\omega''$\\
						$\emptyset$
					\end{tabular}
				\end{center}
				and
				\begin{center}
					\begin{tabular}{c}
						$\omega\omega'$\\
						$\omega''$\\
						$\emptyset$
					\end{tabular}
					\quad
					$\lor$
					\quad
					\begin{tabular}{c}
						$\omega''$\\
						$\omega\omega'$\\
						$\emptyset$
					\end{tabular}
				\end{center}
				but every preorder in the form
				\begin{center}
					\begin{tabular}{c}
						$\omega''$\\
						$\omega\omega'$\\
						$\emptyset$
					\end{tabular}
				\end{center}
				belongs to $\mathcal F_2$.
				
				\item When $P_{i_{k+1}}$ is in the form
				\begin{center}
					
					\begin{tabular}{c}
						$\emptyset$\\
						$\omega$\\
						$\omega_2,\omega'$
					\end{tabular}
					
				\end{center}
				there are two sub cases to consider.
				
				First, suppose without loss of generality that $\omega_2$ and $\omega'$ belong to $L_1(P')$
				then, as $P'$ is not in $\mathcal F_2$, we must have
				\begin{center}
					$P'=$
					\begin{tabular}{c}
						$\omega$\\
						$\emptyset$\\
						$\omega_2\omega'$
					\end{tabular}
				\end{center}
				In this case $\omega$ must be at $L_2(P_{i_k})$.
				If $P_{i_k}$ is linear $\omega_2$ is in $L_1(P_{i_k})$ or $L_3(P_{i_k})$
				and it is in $\mathcal F_2$. Thus suppose $P_{i_k}$ is not linear,
				hence it must have $\omega_2$ and $\omega'$ at the same level
				and wherever we put $\omega_2,\omega'$ we obtain that $P_{i_k}$ is in $\mathcal F_2$.
				
				Suppose now, without loss of generality, that $\omega_2$ is in $L_1(P_{i_k})$ and $\omega'$ is in $L_1(P')$.
				Observe that $P_{i_k}$ cannot be linear in this case and that we already examined the case when
				$\omega_2$ and $\omega'$ are both in $L_1(P_{i_k})$, thus we can suppose $\omega$ and $\omega'$ are at the same level, which
				 must be $L_2$ or $L_3$. If $\omega$ and $\omega'$ are in $L_2(P{i_k})$ then
				\begin{center}
					$P{i_k}=$
					\begin{tabular}{c}
						$\emptyset$\\
						$\omega\omega'$\\
						$\omega_2$
					\end{tabular}
				\end{center}
				which belongs to $\mathcal F_2$.
				If, on the other hand, $\omega$ and $\omega'$ are in $L_3(P{i_k})$ we have that
				\begin{center}
					$P{i_k}=$
					\begin{tabular}{c}
						$\omega\omega'$\\
						$\emptyset$\\
						$\omega_2$
					\end{tabular}
				\end{center}
				thus $\omega$ is in $L_2(P')$ and $P'$ must be in one of the following forms	
				\begin{center}
					\begin{tabular}{c}
						$\omega' $\\
						$\omega$\\
						$\omega_2$
					\end{tabular}
					\qquad
					\begin{tabular}{c}
						$\emptyset$\\
						$\omega\omega' $\\
						$\omega_2$
					\end{tabular}
					\qquad
					\begin{tabular}{c}
						$\emptyset$\\
						$\omega$\\
						$\omega_2\omega' $
					\end{tabular}
				\end{center}
				and all of these forms are in $\mathcal F_2$.
				
				\item Suppose $P_{i_{k+1}}$ has the form
				\[
				\begin{array}{c}
				\omega_2\omega'\\
				\omega\\
				\emptyset
				\end{array}
				\]
				
				In this case we must have $\omega_2,\omega'\in L_3(P_{i_k})$ and $\omega_2,\omega'\in L_3(P')$
				and necessarily either $\omega\in L_2(P_{i_k})$ or $\omega\in L_2(P')$.
				Thus either $P_{i_{k+1}}=P_{i_k}$ or $P_{i_{k+1}}=P'$ which is not possible
				because then $P_{i_{k+1}}$ can be obtained in $k$ steps and it contradicts the inductive hypothesis.

				\item If $P_{i_{k+1}}$ has the form
				\begin{center}
					\begin{tabular}{c}
						$\omega''$\\
						$\omega\omega'$\\
						$\emptyset$
					\end{tabular}
				\end{center}
				then necessarily $L_1(P_{i_{k}})=L_1(P')=\emptyset$ and $\omega''$ is in $L_3(P_{i_{k}})$ and $L_3(P')$.
				If neither $P_{i_{k}}$ nor $P'$ belong to $\mathcal F_2$ the only combination yielding the given configuration is
				\begin{center}
					\begin{tabular}{c}
						$\omega\omega'' $\\
						$\omega'$\\
						$\emptyset$
					\end{tabular}
					$\lor$
					\begin{tabular}{c}
						$\omega'\omega''$\\
						$\omega$\\
						$\emptyset$
					\end{tabular}
				\end{center}
				as at least one of them have $\omega_2$ at level $L_3$,
				then at least one of them is in $\mathcal F_2$.
				
				\item If $P_{i_{k+1}}$ has the form
				\[
				\begin{array}{c}
				\emptyset\\
				\omega\omega'\\
				\omega''
				\end{array}
				\]
				Suppose, without loss of generality, that $\omega,\omega'\in L_2(P')$ then
				necessarily $P'\in\mathcal{F}_2$ since it must be in one of the forms
				\begin{gather*}
				\begin{array}{c}
				\emptyset\\
				\omega\omega'\\
				\omega''
				\end{array}
				\qquad
				\begin{array}{c}
				\emptyset\\
				\omega\omega'\omega''\\
				\emptyset
				\end{array}
				\qquad
				\begin{array}{c}
				\omega''\\
				\omega\omega'\\
				\emptyset
				\end{array}
				\end{gather*}

				Thus we can assume that $\set{\omega,\omega'}\nsubseteq L_2(P')$
				and $\set{\omega,\omega'}\nsubseteq L_2(P_{i_k})$.
				Since $\set{\omega,\omega'}= L_2(P_{i_{k+1}})$
				it must also be the case that
				$\set{\omega,\omega'}\nsubseteq L_3(P')$
				and $\set{\omega,\omega'}\nsubseteq L_3(P_{i_k})$
				because if, for instance, we have that $\set{\omega,\omega'}\subseteq L_3(P')$
				then we need $\set{\omega,\omega'}\subseteq L_2(P_{i_k})$.
				
				We have a case with
				$\omega\in L_3(P_{i_k})$, $\omega'\in L_2(P_{i_k})$,
				$\omega\in L_2(P')$ and $\omega'\in L_3(P')$
				(the other possible case is symmetric)
				but then necessarily either
				$\omega''\in L_1(P_{i_k})$ or $\omega''\in L_1(P')$.
				Suppose $\omega''\in L_1(P')$ (the other case is, again, symmetric)
				thus we have one of the following cases:
				
				One possibility is
				\begin{equation*}
				P_{i_{k+1}}=
				\begin{array}{c}
				\omega'\\
				\omega\\
				\omega''
				\end{array}
				\lor
				\begin{array}{c}
				\omega\\
				\omega'\\
				\omega''
				\end{array}
				\end{equation*}
				in this case, we have at least one linear order with $\omega_2$ at an extreme level, hence
				at least one of the preorders is in $\mathcal{F}_2$.
				
				Another possible configuration is
				\begin{equation*}
				P_{i_{k+1}}=
				\begin{array}{c}
				\omega'\\
				\omega\\
				\omega''
				\end{array}
				\lor
				\begin{array}{c}
				\omega\\
				\omega'\omega''\\
				\emptyset
				\end{array}
				\end{equation*}
				but in this case $P_{i_{k}}\in\mathcal{F}_2$.
				
				The remaining possibility is
				\begin{equation*}
				P_{i_{k+1}}=
				\begin{array}{c}
				\omega'\\
				\omega\\
				\omega''
				\end{array}
				\lor
				\begin{array}{c}
				\omega\omega''\\
				\omega'\\
				\emptyset
				\end{array}
				\end{equation*}
				If $\omega'\neq\omega_2$ and $\omega''\neq\omega_2$ then $\omega=\omega_2$
				and
				\begin{equation*}
				P_{i_{k}}=
				\begin{array}{c}
				\omega_2\omega''\\
				\omega'\\
				\emptyset
				\end{array}
				\in
				\mathcal{F}_2
				\end{equation*}

			\end{enumerate}

		\end{enumerate}
		
	\end{enumerate}
	
%
%
%
%
%
%
%
%
%
%
%

\bibliographystyle{plain}
\bibliography{biblio}


\end{document}